\documentclass[10pt,twocolumn,aps,english,pra,superscriptaddress,showkeys]{revtex4-1}

	\usepackage{amsmath}
    \usepackage{amsfonts}
	\usepackage{epsfig}
	\usepackage[colorlinks,hyperindex]{hyperref}
    \usepackage{color}
    \usepackage{amsthm}
    \usepackage{hyperref}
	\hypersetup
	{
		colorlinks,%
		citecolor=green,%
		linkcolor=blue,%
		urlcolor=blue,%
	}

	\newtheorem{theorem}{Theorem}
    \newtheorem{property}{Property}
    \newtheorem{conjecture}[theorem]{Conjecture}

\usepackage{braket}
\newcommand{\gv}[1]{\ensuremath{\text{\boldmath$ #1 $}}}
\newcommand{\abs}[1]{\left| #1 \right|} 
\let\baraccent=\= 
\renewcommand{\=}[1]{\stackrel{#1}{=}} 

    \newlength{\halfpagewidth}
    \divide\halfpagewidth by 2

\newcommand{\pref}[1]{\hyperref[#1]{Property~\ref{#1}}}
\newcommand{\conjref}[1]{\hyperref[#1]{Conjecture~\ref{#1}}}
\newcommand{\equationref}[1]{\hyperref[#1]{Eq.~(\ref{#1})}}

\begin{document}

\title{State-Independent Error-Disturbance Tradeoff For Measurement Operators}

\author{S. S. Zhou}
\affiliation{Kuang Yaming Honors School, Nanjing University, Nanjing, Jiangsu 210093, China}
\author{Shengjun Wu}
\affiliation{Kuang Yaming Honors School, Nanjing University, Nanjing, Jiangsu 210093, China}
\author{H.\ F. Chau}
\affiliation{Department of Physics, The University of Hong Kong, Pokfulam Road, Hong Kong}

\begin{abstract}

In general, classical measurement statistics of a quantum measurement is disturbed by performing an additional incompatible quantum measurement beforehand.
Using this observation, we introduce a state-independent definition of disturbance by relating it to the distinguishability problem between two classical statistical distributions --- one resulting from a single quantum measurement and the other from a succession of two quantum measurements.
Interestingly, we find an error-disturbance trade-off relation for any measurements in two-dimensional Hilbert space and for measurements with mutually unbiased bases in any finite-dimensional Hilbert space.
This relation shows that error should be reduced to zero in order to minimize the sum of error and disturbance.
We conjecture that a similar trade-off relation can be generalized to any measurements in an arbitrary finite-dimensional Hilbert space.
\end{abstract}


\maketitle

\section{Introduction}
\label{sec:intro}

The uncertainty principle has been regarded as a fundamental principle in quantum mechanics.
It asserts that we cannot get the precise values of two physical observables in a quantum state, unless they are compatible.
The well-known version of this principle was formulated by Heisenberg in 1927, namely,~\cite{heisenberg1927anschaulichen}
\begin{equation}
\label{eq:HUP_0}
\Delta x \Delta p \geq \hbar/2 .
\end{equation}
A more general form of it can be written as
\begin{equation}
\label{eq:HUP}
\varepsilon(A) \eta(B) \geq \frac{\abs{\bra{\psi} [A,B] \ket{\psi}}}{2},
\end{equation}
where $\varepsilon(A)$ is the error with which measurement of operator $A$ is carried out, and $\eta(B)$ is the disturbance on the following measurement of operator $B$ caused by measurement of $A$. In \equationref{eq:HUP_0}, $\Delta x$ and $\Delta p$ can be interpreted as error and disturbance when $A$ and $B$ are position and momentum operators.
Mathematically, \equationref{eq:HUP} comes from the Robertson's uncertainty relation~\cite{robertson1929uncertainty}:
\begin{equation}
\label{eq:RR}
\sigma(A) \sigma(B) \geq \frac{\abs{\bra{\psi} [A,B] \ket{\psi}}}{2},
\end{equation}
where $\sigma(X) \equiv \braket{\psi|X^2|\psi} - \braket{\psi|X|\psi}^2$ is the standard derivation of observable $X$ in a quantum state $\ket{\psi}$.
Note that while \equationref{eq:RR}, usually regarded as a rigorous version of Heisenberg's uncertainty principle~\cite{john1955mathematical,bohm1951quantum,schiff1955quantum}, can be proven mathematically, the justification of relation \equationref{eq:HUP} is currently on hot debate because additional conditions have been used in its derivation \cite{ozawa2014heisenberg}.
More importantly, several experiments showed that \equationref{eq:HUP} is violated~\cite{sulyok2013violation,baek2013experimental,kaneda2014experimental}.
Thus, the trade-off relation that the higher the precision of measuring $A$, the stronger the disturbance on $B$ cannot be well-captured by \equationref{eq:HUP}.

Many important works in this area have been done, but the definitions of error and disturbance are still not settled~\cite{busch2014colloquium}.
Ozawa used noise-operator based error definition and proposed a ``universally valid error-disturbance relation''~\cite{ozawa2003universally}:
\begin{equation}
\varepsilon(A)\eta(B) + \varepsilon(A)\sigma(B) + \sigma(A)\eta(B) \geq \frac{\abs{\bra{\psi} [A,B] \ket{\psi}}}{2}.
\end{equation}
This uncertainty relation was later verified experimentally~\cite{lund2010measuring,rozema2012violation,erhart2012experimental,sulyok2013violation,
baek2013experimental,kaneda2014experimental,ringbauer2014experimental}
and inspired a lot of work on uncertainty relations~\cite{hall2004prior,branciard2013error,branciard2014deriving},
but some shortcomings were also pointed out~\cite{korzekwa2014operational,dressel2014certainty}.
For example, it seems to violate the proposed operational constraint that the error and disturbance should be non-zero if the outcome distribution is deviated from what is expected according to Born's rule~\cite{korzekwa2014operational}.

Using distance between distributions is another way to quantifying measurement errors~\cite{busch2013proof,buscemi2014noise,coles2015state},
Busch \emph{et al.} proved the original Heisenberg's error-disturbance relation in \equationref{eq:HUP} by defining error and disturbance as figures of merit characteristic of the measuring devices~\cite{busch2013proof,busch2014measurement},
which generated a debate over different approaches used in formalizing uncertainty relations~\cite{ozawa2013disproving,busch2014measurement2}.

In this paper, we introduce a straightforward definition of error and disturbance.
The key observation is that given an arbitrary quantum state, the measurement statistics of a measurement operation on that state is unchanged
if and only if we perform an additional compatible measurement to the state beforehand.
Thus, we may define the disturbance of $\mathcal{B}$ (measurement of operator $B$) due to $\mathcal{A}$ (measurement of operator $A$) as the distance between the two probability distributions of the measurement outcomes due to $\mathcal{B}$ and $\mathcal{B}\circ \mathcal{A}$ maximized over all possible input quantum states.
We introduce the definitions of error and disturbance in \autoref{sec:def} and report a few basic properties of these quantities in \autoref{sec:prop}.
Then, in \autoref{sec:relation}, we prove the error-disturbance trade-off relation for the case of $2$-dimensional Hilbert space.
In particular, we derive a sharp lower bound of the sum of error and disturbance.
We also give the trade-off relation in $d$-dimensional Hilbert space for a special but important case.
Finally, we draw a few conclusions in \autoref{sec:conclusion}.

\section{State-Independent Error And Disturbance Of Projective Measurements}
\label{sec:def}

Suppose one is given a density matrix $\rho$ in a $d$-dimensional Hilbert space with $d\ge 2$.
Let $\mathcal{A}$ be the projective measurements of operator $A$ with rank-one projectors.
(Unless otherwise stated, all measurements in this paper are associated with rank-one projectors.  Note that our discussion can be easily extended to the case of a general positive operator-valued measurement.  We restrain from doing so to avoid unnecessary notational and indexing complications.)
The probability distribution obtained from applying $\mathcal{A}$ to $\rho$ is given by the vector
\begin{equation}
 \label{P_A_defined}
 P_\mathcal{A}(\rho) = \big(p^{(\mathcal{A})}_i(\rho) \big)_{i=1}^d \equiv
 \big(\bra{a_i} \rho \ket{a_i} \big)_{i=1}^d ,
\end{equation}
where $\ket{a_i}\bra{a_i}$ is the rank-one projector corresponding to the $i$\textsuperscript{th} measurement outcome.
We now consider measuring $\rho$ using another projective measurement $\mathcal{A}'$ before feeding the resultant state to $\mathcal{B}$.
We write the probability distribution of the measurement outcomes of $\mathcal{A}'$ by $P_{\mathcal{A}'}(\rho)$.
More importantly, the probability distribution of the final measurement outcomes of $\mathcal{B} \circ \mathcal{A}'$ is given by $P_{\mathcal{B}\circ \mathcal{A}'}(\rho) = P_\mathcal{B}(\rho')$ where
$\rho' = \sum_i \bra{a'_i} \rho \ket{a'_i} \ \ket{a'_i} \bra{a'_i}$ with
$\ket{a'_i}\bra{a'_i}$ being the rank-one projector corresponding to the $i$\textsuperscript{th} measurement outcome of ${\mathcal A}'$.

In general, $P_{{\mathcal B} \circ {\mathcal A}'}(\rho)$ is different from $P_{{\mathcal B}}(\rho)$ as measurement changes the state of a quantum system.
We would like to know how a change in the intermediate measurement ${\mathcal A}'$ affects the change of $P_{{\mathcal B}}(\rho)$ through their classical statistics of their measurement outcomes only.
With this motivation in mind, for any given metric $D(\cdot,\cdot)$ of an Euclidean space, we define
the \textbf{state-dependent error} between $P_\mathcal{A}(\rho)$ and $P_{\mathcal{A}'}(\rho)$,
and the \textbf{state-dependent disturbance} between $P_\mathcal{B}(\rho)$ and $P_\mathcal{B}(\rho')$ as
\begin{equation}
    \label{eq:sd-error_def}
    \varepsilon_{\rho} ( \mathcal{A}, \mathcal{A}') = D(P_\mathcal{A}(\rho), P_{\mathcal{A}'}(\rho))
\end{equation}
and
\begin{equation}
   \label{sd-disturbance_def}
    \eta_{\rho}(  \mathcal{A}', \mathcal{B}) = D(P_{\mathcal{B}}(\rho), P_{\mathcal{B}}(\rho')) ,
\end{equation}
respectively.
Since our goal is to study the maximum pointwise deviation in the distribution of the measurement outcomes, we use the metric based on the infinity norm, namely,
\begin{equation}
 \label{eq:inf-norm}
 D(x,y) = \max_i |x_i - y_i| .
\end{equation}

We now define the \textbf{state-independent error} and
\textbf{state-independent disturbance} by
\begin{equation}
    \label{sind-error_def}
    \varepsilon(\mathcal{A}, \mathcal{A}') = \max_\rho \ \varepsilon_{\rho}(  \mathcal{A}, \mathcal{A}')
\end{equation}
and
\begin{equation}
 \label{sind-disturbance_def}
    \eta(\mathcal{A}', \mathcal{B}) = \max_\rho \ \eta_{\rho}(  \mathcal{A}', \mathcal{B}) .
\end{equation}
From now on, the terms ``error'' and ``disturbance'' refer to the state-independent versions unless otherwise stated.
Note that these definitions meet the proposed operational constraint \cite{korzekwa2014operational} for $\varepsilon(\mathcal{A}, \mathcal{A}') = 0$ if and only if $\mathcal{A} = \mathcal{A}'$ and
$\eta(\mathcal{A}', \mathcal{B}) = 0$ if and only if $\mathcal{A}' = \mathcal{B}$.

Finally,
to obtain a trade-off relation between error and disturbance in one measurement, that is,
to find out how much we need to sacrifice on one to lower the other, just as what Heisenberg did,
we introduce the \textbf{state-independent overall error}
\begin{equation}
    \label{sind-overall-error_def}
    \Delta(\mathcal{A}, \mathcal{A}', \mathcal{B}) = \max_\rho \left( \varepsilon_{\rho}(  \mathcal{A}, \mathcal{A}') + \eta_{\rho}(  \mathcal{A}', \mathcal{B})\right).
\end{equation}
Clearly, $\varepsilon + \eta \geq \Delta$.

\section{Basic Properties Of State-Independent Error And Disturbance}
\label{sec:prop}

Since
\begin{equation}
\begin{gathered}
\varepsilon_{\rho}(\mathcal{A}, \mathcal{A}') = \max_i \abs{\mathrm{tr} \left( \rho ( \ket{a_i}\bra{a_i} - \ket{a_i'}\bra{a_i'} ) \right)}\\
\eta_{\rho}(\mathcal{A}', \mathcal{B}) = \max_i \big| \mathrm{tr} \big( \rho ( \ket{b_i}\bra{b_i} - \sum_j \ket{a_j'} |\braket{a_j' | b_i}|^2 \bra{a_j'} )\big) \big|,
\end{gathered}
\end{equation}
we have
\begin{equation}
\label{eq:err-spec_rad}
\varepsilon(\mathcal{A},\mathcal{A}') = \max_i \ R \left( \ket{a_i}\bra{a_i} - \ket{a_i'}\bra{a_i'}\right)
\end{equation}
and
\begin{equation}
\label{eq:disturbance-spec_rad}
\eta(\mathcal{A}',\mathcal{B}) = \max_i \ R \big(\ket{b_i}\bra{b_i} - \sum_k \abs{\braket{b_i | a_k'}}^2 \ket{a_k'} \bra{a_k'} \big) .
\end{equation}
Here, $R(\cdot)$ is the spectral radius of a matrix (the largest of absolute values of the eigenvalues).
Similarly, we have
\begin{align}
\Delta(\mathcal{A},\mathcal{A}',\mathcal{B}) =& \max_{i,j,\pm}
\ R \big( \ket{a_i}\bra{a_i} - \ket{a_i'}\bra{a_i'} \pm \nonumber \\
&\ket{b_j}\bra{b_j} \mp \sum_k \abs{\braket{b_j | a_k'}}^2 \ket{a_k'} \bra{a_k'} \big).
\end{align}
Note that the maximum of $\varepsilon_\rho$ can be attained by a pure state $\rho$; and similarly for $\eta_\rho$ and $\varepsilon_\rho+\eta_\rho$.

\begin{property}[Maximum Error and Disturbance]
\label{prop:pp1} The error satisfies
\begin{equation}
\varepsilon(\mathcal{A},\mathcal{A}') \leq 1
\end{equation}
for any ${\mathcal A}, {\mathcal A}'$, with equality if and only if $\braket{a_i'|a_i} = 0$ for some $i$.
In addition, the disturbance obeys
\begin{equation}
\label{eq:disturbance_upper_bound}
\eta(\mathcal{A}', \mathcal{B}) \leq 1 - 1/d;
\end{equation}
for any ${\mathcal A}', {\mathcal B}$, with equality if and only if there is an $i$ such that $\ket{b_i}$ is unbiased in $\{\ket{a_j'}\}_{j=1}^d$.
That is to say, $|\braket{a'_j|b_i}|^2 = 1/d$ for any $j$.
\end{property}
\begin{proof}
The rank of the matrix $\ket{a_i}\bra{a_i} - \ket{a_i'}\bra{a_i'}$ is at most~2.
Thus, the spectral radius of this matrix can be calculated easily as $\ket{a_i}$ and $\ket{a'_i} - \braket{a_i|a'_i} \ket{a_i}$ are orthogonal.
Hence, \equationref{eq:err-spec_rad} becomes
\begin{equation}
\label{eq:expression_error}
\varepsilon(\mathcal{A},\mathcal{A}') = \max_i \sqrt{1 - |\braket{a_i'|a_i}|^2}.
\end{equation}
Consequently, $\varepsilon(\mathcal{A},\mathcal{A}') \le 1$ with equality holds when there exists $i$ such that $\braket{a_i'|a_i} = 0$.

\equationref{eq:disturbance-spec_rad} can be rewritten as
\begin{equation}
\eta(\mathcal{A}',\mathcal{B}) = \max_{\psi,i} \big| \abs{\braket{\psi | b_i}}^2
- \sum_k \abs{\braket{\psi | a_k'}}^2\abs{\braket{b_i | a_k'}}^2 \big| .
\end{equation}
Note that the spectral radius of the matrix $\ket{b_i}\bra{b_i} - \sum_k \abs{\braket{b_i | a_k'}}^2 \ket{a_k'} \bra{a_k'} $ is unchanged by multiplying a phase $e^{i\delta}$ to $\ket{a_k'}$.
Hence, we may always set each $\braket{b_i|a_k'}$ to be real and non-negative.
Then by the Perron-Frobenius theorem~\cite{bapat1997nonnegative}, there is always a positive eigenvalue corresponding to the spectral radius.  In other words,
\begin{equation}
\label{eq:perron}
\eta(\mathcal{A}',\mathcal{B}) = \max_{\psi,i} \abs{\braket{\psi | b_i}}^2
- \sum_k \abs{\braket{\psi | a_k'}}^2\abs{\braket{b_i | a_k'}}^2.
\end{equation}
Using the Cauchy-Schwarz inequality,
\begin{align}
\abs{\braket{\psi | b_i}}^2 \leq&  \Big( \sum_k \abs{\braket{\psi | a_k'}}\abs{\braket{a_k'| b_i }} \Big)^2 \nonumber \\
\leq& \sum_k \abs{\braket{\psi | a_k'}}^2\abs{\braket{b_i | a_k'}}^2 \sum_k 1^2.
\end{align}
So $\abs{\braket{\psi | a_k'}}^2\abs{\braket{b_i | a_k'}}^2 \geq \abs{\braket{\psi | b_i}}^2/d$.
Hence, \equationref{eq:disturbance_upper_bound} is true.
Moreover, the equality holds if and only if there exists an $i$ such that $\abs{\braket{\psi|b_i}} = 1$ and $\abs{\braket{b_i | a_1'}}^2 = \abs{\braket{b_i | a_2'}}^2 = \cdots = \abs{\braket{b_i | a_d'}}^2$.
In other words, $\ket{b_i}$ is unbiased in $\{\ket{a_k'}\}_{k=1}^{d}$.
\end{proof}

\pref{prop:pp1} can be understood as follows.
Suppose we take $\rho = \ket{a_i}\bra{a_i}$, the measurement operation $\mathcal{A}$ will give a post-measurement state $\ket{a_i}$ for sure.
In contrast, when we perform the measurement $\mathcal{A'}$, we have no chance of getting the state $\ket{a'_i}$ if $\ket{a_i}$ and $\ket{a'_i}$ are orthogonal.
This is a situation in which the error between the two measurement outcomes of $\mathcal{A}$ and $\mathcal{A'}$ on the same state $\rho$ is maximized.
In addition, suppose there is an $i$ such that $\ket{b_i}$ is unbiased in $\{\ket{a'_j}\}_{j=1}^d$.
Then by choosing $\rho = \ket{b_i}\bra{b_i}$, the measurement operation $\mathcal{A'}$ will completely erase the information on the distribution of $\rho$ in $\{\ket{b_j}\}_{j=1}^d$.  Hence, the disturbance $\eta({\mathcal A}',{\mathcal B})$ is maximized.

\begin{property}[Reducibility]
\label{pp2}
Suppose $A$, $A'$ and $B$ are operators on a Hilbert space $\gv{H} = \gv{H}_1\oplus\gv{H}_2\oplus\cdots\oplus\gv{H}_l$
such that $\ket{a_i}$'s and $\ket{a'_i}$'s are in $\gv{H}_1\cup\gv{H}_2\cup\cdots\cup\gv{H}_l$, then
\begin{equation}
\label{eq:pp2_error}
\varepsilon(\mathcal{A},\mathcal{A}') = \max_{1 \leq k \leq l} \varepsilon(\mathcal{A}|_{\gv{H}_k}, \mathcal{A}'|_{\gv{H}_k}).
\end{equation}
If we substitute $\ket{b_i}$'s for $\ket{a_i}$'s in above conditions, we have similarly
\begin{equation}
\label{eq:pp2_disturbance}
\eta(\mathcal{A}',\mathcal{B}) = \max_{1 \leq k \leq l} \eta(\mathcal{A}'|_{\gv{H}_k}, \mathcal{B}|_{\gv{H}_k}).
\end{equation}
\end{property}
\begin{proof}
We only prove \equationref{eq:pp2_disturbance} as the proof of
\equationref{eq:pp2_error} is similar.
Let $\gv{B}_k = \mathrm{Span}\{\ket{b_i} \colon \ket{b_i}\in \gv{H}_k\}$.
Obviously, $\gv{B}_k \subseteq \gv{H}_k$.
And since $\gv{H} = \gv{B}_1\oplus\gv{B}_2\oplus\cdots\oplus\gv{B}_l$, we know that $\gv{H}_k = \gv{B}_k$.
Therefore, for any $\ket{b_i}, \ket{a'_j}$ not in same $\gv{H}_k$, $\braket{b_i | a'_j} = 0$.
Suppose $\ket{b_i} \in \gv{B}_{k_i} = \gv{H}_{k_i}$, \equationref{eq:disturbance-spec_rad} becomes
\begin{align}
& \eta(\mathcal{A}',\mathcal{B}) \nonumber \\
={}&\max_i R \big(\ket{b_i}\bra{b_i} - \sum_{\ket{a'_j} \in \gv{H}_{k_i}} \abs{\braket{b_i | a'_j}}^2 \ket{a'_j} \bra{a'_j} \big) \nonumber \\
={}&\max_{\substack{k ,\; i \in {\{i'|k_{i'} = k\}}}} R \big(\ket{b_i}\bra{b_i} - \sum_{\ket{a'_j} \in \gv{H}_{k_i}} \abs{\braket{b_i | a'_j}}^2 \ket{a'_j} \bra{a'_j} \big) \nonumber \\
={}&\max_k \eta(\mathcal{A}'|_{\gv{H}_{k}},\mathcal{B}|_{\gv{H}_{k}}) .
\end{align}
Thus, \equationref{eq:pp2_disturbance} holds.
\end{proof}
\begin{property}[Subsystems]
\label{pp2_prime}
Suppose $A$, $A'$ and $B$ are operators on the Hilbert space $\gv{H} = \gv{H}_1\otimes\gv{H}_2\otimes\cdots\otimes\gv{H}_l$, then
\begin{equation}
\label{eq:pp2_prime_error}
\varepsilon(\mathcal{A},\mathcal{A}') \geq \max_{1 \leq k \leq l} \varepsilon(\mathcal{A}|_{\gv{H}_k}, \mathcal{A}'|_{\gv{H}_k})
\end{equation}
and
\begin{equation}
\label{eq:pp2_prime_disturbance}
\eta(\mathcal{A}',\mathcal{B}) \geq \max_{1 \leq k \leq l}  \eta(\mathcal{A}'|_{\gv{H}_k}, \mathcal{B}|_{\gv{H}_k}) .
\end{equation}
\end{property}
\begin{proof}
Suppose $l=2$.
Then, $\ket{a'_j}$ can be represented by $\ket{{a'}^{1}_{m'}} \otimes \ket{{a'}^{2}_{n'}}$,
where $\ket{{a'}^{1}_{m'}} \in {\gv{H}_1}$ and $\ket{{a'}^{2}_{n'}} \in {\gv{H}_2}$.
Similarly, we write $\ket{b_i} = \ket{b^{1}_m} \otimes \ket{b^{2}_n}$.
Then,
\begin{align}
& R \big(\ket{b_i}\bra{b_i} - \sum_j \abs{\braket{b_i | a'_j}}^2 \ket{a'_j} \bra{a'_j} \big) \nonumber \\
\geq{}& \max_{\ket{\psi_1}, \ket{\psi_2}} \abs{\braket{\psi_1|b^1_m}}^2 \abs{\braket{\psi_2|b^2_n}}^2
 \nonumber \\
& \quad - \sum_{m',n'} \abs{\braket{\psi_1|{a'}^1_{m'}}}^2 \abs{\braket{\psi_2|{a'}^2_{n'}}}^2 \abs{\braket{b^1_m|{a'}^1_{m'}}}^2 \abs{\braket{b^2_n|{a'}^2_{n'}}}^2 \nonumber \\
\geq{}& \max_{\ket{\psi_1}} \abs{\braket{\psi_1|b^1_m}}^2
- \sum_{m',n'} \abs{\braket{\psi_1|{a'}^1_{m'}}}^2 \abs{\braket{b^2_n|{a'}^2_{n'}}}^4 \abs{\braket{b^1_m|{a'}^1_{m'}}}^2 \nonumber \\
\geq{}& \max_{\ket{\psi_1}} \abs{\braket{\psi_1|b^1_m}}^2 - \sum_{m'} \abs{\braket{\psi_1|{a'}^1_{m'}}}^2 \abs{\braket{b^1_m|{a'}^1_{m'}}}^2
\nonumber \\
\geq{}& R \big(\ket{b^1_m}\bra{b^1_m} - \sum_{m'} \abs{\braket{b^1_m | {a'}^1_{m'}}}^2 \ket{{a'}^1_{m'}} \bra{{a'}^1_{m'}} \big).
\end{align}
We conclude that $\eta(\mathcal{A},\mathcal{B}) \geq \eta(\mathcal{A}|_{{\gv{H}_1}}, \mathcal{B}|_{{\gv{H}_1}})$.
By induction of $l$, we show the validity of \equationref{eq:pp2_prime_disturbance}.
\equationref{eq:pp2_prime_error} can be proven in a similar way.
\end{proof}

\begin{property}[Upper bound Estimation]
\label{pp3}
\begin{equation}
\label{eq:upper1}
\eta(\mathcal{A}',\mathcal{B}) \leq \max_i \left[\Big(1 - \frac{1}{d}\Big)\Big(1 - \sum_k \abs{\braket{a'_k|b_i}}^4\Big)\right]^{1/2}.
\end{equation}
The equality holds when $d = 2$. When $d > 2$, the equality holds if and only if there exists an $i$ such that $\ket{b_i}$ is unbiased in $\{\ket{a'_j}\}_{j=1}^d$, or $\mathcal{A}' = \mathcal{B}$.
Furthermore,
\begin{equation}
\label{eq:upper2}
\eta(\mathcal{A}',\mathcal{B}) \leq \max_{i,j} \abs{\braket{a'_j|b_i}} \sum_{k \neq j} \abs{\braket{a'_k|b_i}}.
\end{equation}
The equality holds when $d = 2$. When $d > 2$, the equality holds if and only if for an $i$ which maximizes $\max_{j} \abs{\braket{a'_j|b_i}} \sum_{k \neq j} \abs{\braket{a'_k|b_i}}$, for all $j,k$ such that $\braket{a'_{j,k}|b_i} \neq 0$, $\abs{\braket{a'_j|b_i}} = \abs{\braket{a'_k|b_i}}$.
\end{property}
\begin{proof}
We first prove \equationref{eq:upper1}.
Suppose the eigenvalues of the matrix
$M = \ket{b_i}\bra{b_i} - \sum_{j=1}^d \abs{\braket{b_i | a'_j}}^2 \ket{a'_j} \bra{a'_j}$ be $\lambda_j$'s.
Arrange these eigenvalues so that
$|\lambda_1| \leq |\lambda_2| \leq \cdots \leq |\lambda_d |= \eta(\mathcal{A}',\mathcal{B})$.
Since $\mathrm{tr}(M) = 0$, we have $ \sum_{i=1}^d \lambda_i = 0$.
Moreover,
\begin{equation}
\begin{aligned}
\mathrm{tr}(M^2) =& \lambda_d^2 + \sum_{i=1}^{d-1} \lambda_i^2 \geq
\lambda_d^2 + \frac{1}{d-1} \bigg( \sum_{i=1}^{d-1} \lambda_i \bigg)^2
 \\
=& \frac{d \lambda_d^2}{d-1} ,
\end{aligned}
\end{equation}
and $\mathrm{tr}(M^2) = 1 - \sum_k \abs{\braket{a'_k|b_i}}^4$.
Hence, \equationref{eq:upper1} is true.
Moreover, the equality holds if and only if $\lambda_1 = \lambda_2 = \ldots = \lambda_{d-1} = - \lambda_{d}/(d-1)$.
A straightforward calculation shows that if $\ket{b_i}$ is unbiased in $\{\ket{a'_j}\}_{j=1}^d$ or there exists a $j$ such that $|\braket{a'_j|b_i}| = 1$, the above equations hold.

Next, we prove \equationref{eq:upper2}.
According to the Frobenius theorem~\cite{butler2013sharp,derzko1965bounds},
\begin{equation}
|\lambda_d| \leq \max_j \sum_{k=1}^d |M_{kj}| = \max_j \abs{\braket{a'_j|b_i}} \sum_{k \neq j} \abs{\braket{a'_k|b_i}}.
\end{equation}
And this equality holds if and only if $d=2$, or for all $j,k$ such that $\braket{a'_{j,k}|b_i} \neq 0$, $\abs{\braket{a'_j|b_i}} = \abs{\braket{a'_k|b_i}}$.
\end{proof}

From \equationref{eq:sd-error_def}--\equationref{eq:inf-norm}, we have
\begin{equation}
\eta^c(\mathcal{A}', \mathcal{B}) \equiv \max_{\rho=\ket{b_i}\bra{b_i}} \eta_\rho(\mathcal{A}', \mathcal{B}) = \max_i \Big( 1 - \sum_k \abs{\braket{a'_k|b_i}}^4 \Big),
\end{equation}
where $\eta^c(\mathcal{A}',\mathcal{B})$ is the so-called calibration disturbance~\cite{busch2014colloquium},
which measures the disturbance with respect to eigenvectors of $B$ as original states.
Moreover, it can also be interpreted as the maximum among probabilities that an original eigenstate $\ket{b_i}$ evolves into
another eigenstate $\ket{b_j}(j\neq i)$ after $\mathcal{B} \circ \mathcal{A'}$.
Calibration disturbance is another way to define a state-independent disturbance~\cite{busch2013proof}.
We already know from the definitions of worst-case disturbance and calibration disturbance obeys
$\eta^c(\mathcal{A}',\mathcal{B}) \leq \eta(\mathcal{A}',\mathcal{B})$.
What \equationref{eq:upper1} shows is that
\begin{equation}
\eta(\mathcal{A}',\mathcal{B}) \leq \left[\Big(1 - \frac{1}{d}\Big) \cdot \eta^c(\mathcal{A}', \mathcal{B}) \right]^{1/2}.
\end{equation}
This relation provides an upper bound of worst-case disturbance $\eta(\mathcal{A}',\mathcal{B})$ -- the geometric mean of $\eta^c(\mathcal{A}',\mathcal{B})$ and the maximum possible disturbance for measurements in $d$-dimensional Hilbert space, namely, $1-1/d$.
Incidentally, for $\varepsilon(\mathcal{A},\mathcal{A}')$, we have
\begin{equation}
\varepsilon(\mathcal{A},\mathcal{A}') = \left[\varepsilon^c(\mathcal{A},\mathcal{A}')\right]^{1/2},
\end{equation}
where $\varepsilon^c(\mathcal{A}',\mathcal{A})$ is defined analogously.

In \autoref{fig:estimation}, we show performance of these two estimations for the case of $d=3$.
We find that the two upper bounds have their own advantages in different situations.

\begin{figure}[ht]
    \centering
    \includegraphics[width=3.2in]{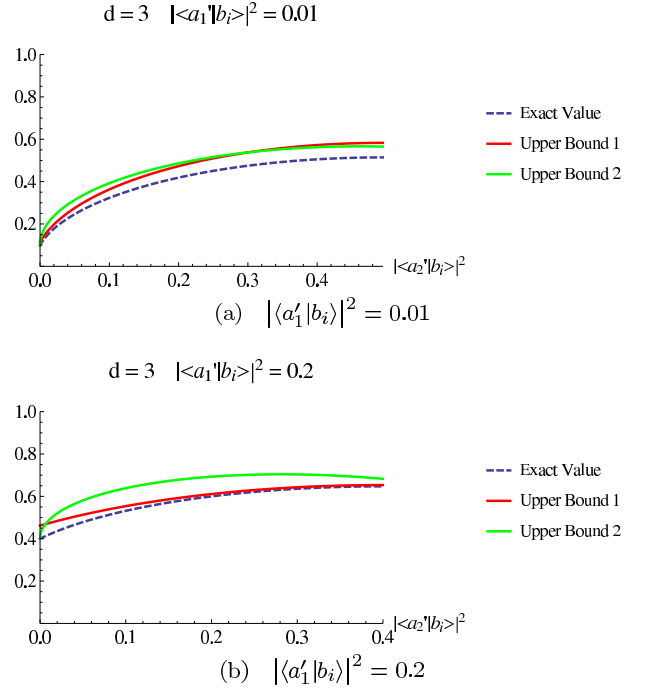}\\
    \caption{(color online) Performance of the upper bound estimation in $d=3$ case: We fix the value of $|\braket{a'_1 | b_i}|^2$,
    and plot the value (before choosing maximum among different $i$'s) of $\eta(\mathcal{A}',\mathcal{B})$  and the two upper bounds
    (``Upper Bound 1'' corresponds to \equationref{eq:upper1} and ``Upper Bound 2'' corresponds to \equationref{eq:upper2}),
    choosing $|\braket{a'_2 | b_i}|^2$ as independent variable.
    In this figure, we assume $|\braket{a'_1 | b_i}|^2 \leq |\braket{a'_2 | b_i}|^2 \leq |\braket{a'_3 | b_i}|^2$.}
    \label{fig:estimation} 
\end{figure}

\section{Trade-off Relation}
\label{sec:relation}

We now report the trade-off relation between error and disturbance.
\begin{theorem}[Trade-off in Two-Dimensional Hilbert space]
\label{thm:thm1}
When $d=2$, the following trade-off between error and disturbance holds:
\begin{equation}
\varepsilon(\mathcal{A},\mathcal{A}') + \eta(\mathcal{A}',\mathcal{B}) \geq \eta(\mathcal{A},\mathcal{B}),
\end{equation}
\begin{equation}
\Delta(\mathcal{A},\mathcal{A}',\mathcal{B}) \geq \eta(\mathcal{A},\mathcal{B}),
\end{equation}
where $\eta(\mathcal{A},\mathcal{B}) = \frac{1}{2}\sin\theta$.
Moreover, the equalities are both attained when $\ket{a_i'} = \ket{a_i}$ for all $i$.
Here $\theta = \arccos \abs{\gv{a} \cdot \gv{b}}$ with $\gv{a},\gv{b}$ being the Bloch vectors of eigenstates of $A$ and $B$.
\end{theorem}
\begin{proof}
We use Bloch vectors to represent density matrix and different bases.  That is,
\begin{gather}
\rho = \frac{I + { {\gv{n}}}\cdot\gv{\sigma}}{2},
\quad \ket{a_i}\bra{a_i}_{i=1,2} = \frac{I \pm { {\gv{a}}}\cdot\gv{\sigma}}{2},
 \notag \\
\quad \ket{b_i}\bra{b_i}_{i=1,2} = \frac{I \pm { {\gv{b}}}\cdot\gv{\sigma}}{2},
\quad \ket{a_i'}\bra{a_i'}_{i=1,2} = \frac{I \pm {{\gv{a'}}}\cdot\gv{\sigma}}{2},
\end{gather}
where $\gv{a},\gv{a'},\gv{b}$ are unit vectors on Block sphere and $\gv{n}$ is in/on Bloch sphere ($\abs{\gv{n}}\leq 1$).
Straightforward calculations gives
\begin{equation}
\begin{aligned}
\varepsilon_\rho(\mathcal{A},\mathcal{A}') =& \frac{1}{2}\abs{(\gv{a} - \gv{a'})\cdot \gv{n}}\\
\leq& \frac{1}{2}\abs{\gv{a} - \gv{a'}} =  \varepsilon(\mathcal{A},\mathcal{A}')
\end{aligned}
\end{equation}
and
\begin{equation}
\begin{aligned}
\eta_\rho(\mathcal{A}',\mathcal{B}) =& \frac{1}{2}\abs{(\gv{b} - \gv{b}\cdot\gv{a'}\gv{a'}) \cdot \gv{n}}  \\
\leq &\frac{1}{2}\abs{\gv{b} - \gv{b}\cdot\gv{a'}\gv{a'}} =  \eta(\mathcal{A}',\mathcal{B}) .
\end{aligned}
\end{equation}
On one hand,
\begin{align}
&\eta(\mathcal{A}',\mathcal{B}) + \varepsilon(\mathcal{A},\mathcal{A}') \nonumber \\
={}& \frac{1}{2}\abs{\gv{a} - \gv{a'}}+\frac{1}{2}\abs{\gv{b} - \gv{b}\cdot\gv{a'}\gv{a'}} \nonumber \\
\geq{}& \frac{1}{2}\abs{\gv{b}\cdot\gv{a'} \gv{a} - \gv{b}\cdot\gv{a'} \gv{a'}}+\frac{1}{2}\abs{\gv{b} - \gv{b}\cdot\gv{a'}\gv{a'}} \nonumber \\
\geq{}& \min_x \frac{1}{2}\abs{x \gv{a} - \gv{b}} = \frac{1}{2}\abs{(1 - \gv{a}\gv{a})\cdot\gv{b}} = \frac{1}{2}\abs{\gv{a}\times\gv{b}}
\end{align}
and
\begin{align}
&\Delta(\mathcal{A},\mathcal{A}',\mathcal{B}) \nonumber \\
={}& \max_\gv{n} ( \frac{1}{2}\abs{(\gv{a} - \gv{a'})\cdot \gv{n}} + \frac{1}{2}\abs{(\gv{b} - \gv{b}\cdot\gv{a'}\gv{a'}) \cdot \gv{n}}) \nonumber \\
\geq{}& \max_\gv{n} \frac{1}{2}\abs{\gv{b}\cdot\gv{a'}}\frac{1}{2}\abs{( \gv{a} -  \gv{a'}) \cdot \gv{n}}+\frac{1}{2}\abs{(\gv{b} - \gv{b}\cdot\gv{a'}\gv{a'})\cdot \gv{n}} \nonumber \\
\geq{}& \frac{1}{2}\abs{\gv{b} - \gv{b}\cdot\gv{a'}\gv{a}} \geq \frac{1}{2}\abs{(1 - \gv{a}\gv{a})\cdot\gv{b}} = \frac{1}{2}\abs{\gv{a}\times\gv{b}}.
\end{align}
On the other hand, when $\gv{a' }= \gv{a}$, $\eta(\mathcal{A}',\mathcal{B}) + \varepsilon(\mathcal{A},\mathcal{A}') =\Delta(\mathcal{A},\mathcal{A}',\mathcal{B})= \frac{1}{2}\abs{\gv{b} - \gv{b}\cdot\gv{a}\gv{a}}$.
Therefore, $\min_{\mathcal{A}'}\eta(\mathcal{A}',\mathcal{B}) + \varepsilon(\mathcal{A},\mathcal{A}') = \min_{\mathcal{A}'} \Delta(\mathcal{A},\mathcal{A}',\mathcal{B}) = \frac{1}{2}\abs{\gv{a}\times\gv{b}}$.
\end{proof}

This result gives a lower bound of the overall error and the sum of error and disturbance,
and therefore provide a version of uncertainty principle in the case of two-dimensional Hilbert space.
We could make $\sin \theta = 1$ by choosing $\gv{a}\cdot\gv{b}=0$.
This means that the overall error is the largest
when $\{\ket{a_i}\}_{i=1}^d$ and $\{\ket{b_i}\}_{i=1}^d$ are mutually unbiased.

\autoref{fig:sketch} shows the performance of the version of uncertainty principle in \autoref{thm:thm1}.
It shows that the values of $\varepsilon(\mathcal{A},\mathcal{A}') + \eta(\mathcal{A}',\mathcal{B})$ and $\Delta(\mathcal{A},\mathcal{A}',\mathcal{B})$ do not change smoothly because we have absolute values, maxima and minima in the definitions of error and disturbance.
In particular, we find that $\varepsilon(\mathcal{A},\mathcal{A}') + \eta(\mathcal{A}',\mathcal{B})$ and $\Delta(\mathcal{A},\mathcal{A}',\mathcal{B})$ are not differentiable near $\mathcal{A}'=\mathcal{A}$.

On the other hand, it turns out that the minima of the overall error $\Delta$
and the sum of error and disturbance are both attained when $\{\ket{a'_i}\}_{i=1,2}$ are exactly $\{\ket{a_i}\}_{i=1,2}$,
regardless of the choice of $\{\ket{b_i}\}_{i=1,2}$.
This is somewhat unexpected for one would guess
that $\mathcal{A}'$ should be somewhere ``between'' $\mathcal{A}$ and $\mathcal{B}$ in order to minimize the overall error.
However, according to \autoref{thm:thm1},
the best approach to minimize the overall error in $2$-dimensional Hilbert space during the consecutive measurement process
is simply to leave the first measurement device unchanged.

Another noteworthy observation is that when $d=2$
\begin{equation}
\min_{\mathcal{A}'} \big(\varepsilon(\mathcal{A},\mathcal{A}') + \eta(\mathcal{A}',\mathcal{B})\big)=\min_{\mathcal{A}'} \Delta(\mathcal{A},\mathcal{A}',\mathcal{B}),
\end{equation}
Because when $\ket{a_i} = \ket{a'_i}$ for all $i$,
$\Delta(\mathcal{A},\mathcal{A},\mathcal{B}) = \eta(\mathcal{A},\mathcal{B})$ and $\varepsilon(\mathcal{A},\mathcal{A}) = 0$,
revealing that it is no easier to reduce the overall error than the sum of error and disturbance.

\begin{figure}[ht]
    \centering
    \includegraphics[width=3.2in]{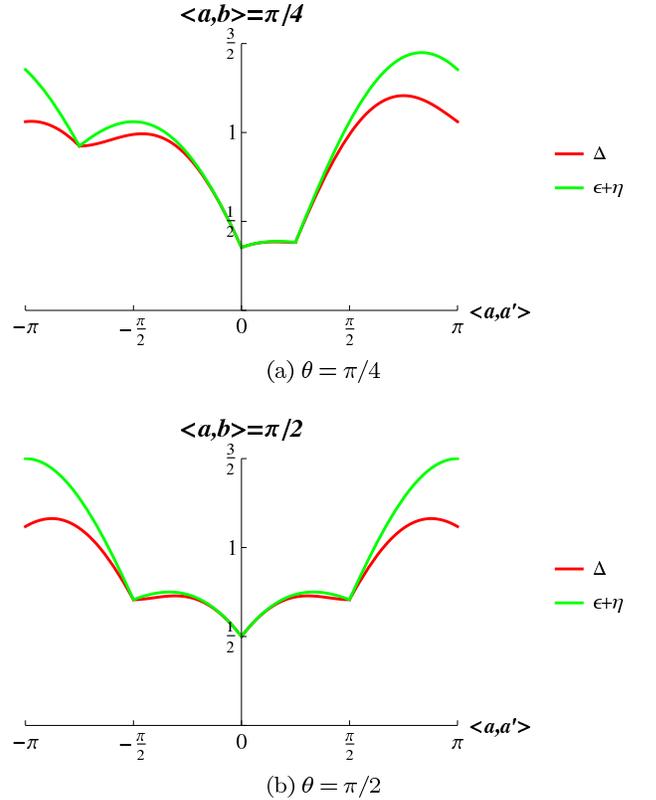}\\
    \caption{(color online) We plot the value of $\varepsilon(\mathcal{A},\mathcal{A}') + \eta(\mathcal{A}',\mathcal{B})$ and $\Delta(\mathcal{A},\mathcal{A}',\mathcal{B})$ respectively
    when $\gv{a'}$ is in $\mathrm{Span}\{\gv{a}, \gv{b}\}$.
    Independent variable $\langle\gv{a},\gv{a'}\rangle$ denotes the angle between $\gv{a}$ and $\gv{a'}$.
    This figure shows clearly that when $\gv{a'} = \gv{a}$, $\varepsilon(\mathcal{A},\mathcal{A}') + \eta(\mathcal{A}',\mathcal{B})$ and $\Delta(\mathcal{A},\mathcal{A}',\mathcal{B})$ are the smallest.}
    \label{fig:sketch}
\end{figure}

Unfortunately, the absence of Bloch representation means that the proof techniques in \autoref{thm:thm1} cannot be used to tackle the trade-off between error and disturbance when $d > 2$.
However, the same conclusion in \autoref{thm:thm1} still works in some special but interesting cases.

\begin{theorem}[Trade-off for Mutually Unbiased Bases in $d$-dimensional Hilbert space]
\label{thm:thm2}
For all $d \geq 2$, if eigenstates of $A$ and $B$ form a pair of mutually unbiased bases (MUB), then we have
\begin{equation}
\varepsilon(\mathcal{A},\mathcal{A}') + \eta(\mathcal{A}',\mathcal{B}) \geq \eta(\mathcal{A},\mathcal{B}) = 1 - 1/d.
\end{equation}
When $\{\ket{a_i'}\}_{i=1}^d$ = $\{\ket{a_i}\}_{i=1}^d$, the equality holds.
\end{theorem}
\begin{proof}
According to \equationref{eq:expression_error}, we need to prove the following inequality:
\begin{align}
& \max_j \sqrt{1 - \abs{\braket{a_j|a_j'}}^2} \nonumber \\
& + \max_i R \big(\ket{b_i}\bra{b_i} - \sum_k \abs{\braket{b_i| a_k'}}^2\ket{a_k'}\bra{a_k'} \big) \nonumber \\
\geq{} &
\max_i R \big(\ket{b_i}\bra{b_i} - \sum_k \abs{\braket{b_i| a_k}}^2 \ket{a_k}\bra{a_k} \big).
\end{align}
Using the triangle inequality for spectral radius~\cite{marshall2010inequalities},
the above inequality is true if we can prove that for any $\ket{b_i}$,
\begin{align}
& \max_j \sqrt{1 - \abs{\braket{a_j|a_j'}}^2} \nonumber \\
\geq{} & R \big(\sum_k  \abs{\braket{b_i| a_k}}^2 \ket{a_k}\bra{a_k} - \sum_k  \abs{\braket{b_i| a_k'}}^2 \ket{a_k'}\bra{a_k'} \big) .
\end{align}
Since $\abs{\braket{b_i| a_k}}^2 = 1/d$, the LHS of the above inequality equals to $\max_j \big| 1/d - \abs{\braket{b_i| a_j'}}^2\big|$.
It is then sufficient to prove that
\begin{equation}
\abs{1/d - \abs{\braket{b_i| a_j'}}^2} \le \sqrt{1 - \abs{\braket{a_j|a_j'}}^2}
\end{equation}
for any $j$.
Using $\sum_k \abs{\braket{a_k | a_j'}}^2 = 1$ and $|\braket{b_i| a_k}|^2 = 1/d$,
it is clear that the value of $\abs{\braket{b_i| a_j'}}^2 = \big| \sum_k \braket{b_i| a_k}\braket{a_k| a_j'} \big|^2$
is between
\[
\frac{1}{d} \big( 1 \pm \sum_{k\neq l} \big|\braket{a_k| a_j'}\big|\cdot\big| \braket{a_l| a_j'} \big| \big).
\]
Hence, all we need to prove is that
\begin{equation}
\label{eq:temp_eq}
\frac{1}{d} \sum_{k\neq l} \big|\braket{a_k| a_j'}\big|\cdot\big| \braket{a_l| a_j'} \big| \le \sqrt{1 - \abs{\braket{a_j|a_j'}}^2}.
\end{equation}
Set $x = 1 - \abs{\braket{a_j|a_j'}}^2$. If $x=0$, \equationref{eq:temp_eq} is obviously true. Otherwise,
\begin{align}
&\sum_{k\neq l} \big|\braket{a_k| a_j'}\big|\cdot\big| \braket{a_l| a_j'} \big| \nonumber \\
={}& 2\sqrt{1 - x} \sum_{m \neq j} \big| \braket{a_m| a_j'} \big| + \sum_{\substack{k\neq l\\k,l \neq j}} \big|\braket{a_k| a_j'}\big|\cdot\big| \braket{a_l| a_j'} \big| \nonumber \\
\leq{}& 2\sqrt{(1 - x)(d-1)} \sqrt{\sum_{m \neq j} \big| \braket{a_m| a_j'} \big|^2} \nonumber \\&
\qquad\qquad\qquad + \sum_{\substack{k\neq l\\k,l \neq j}} \frac{\big|\braket{a_k| a_j'}\big|^2+\big| \braket{a_l| a_j'}\big|^2}{2} \nonumber \\
\leq{}& 2\sqrt{(1 - x)(d-1)} \sqrt{x} + x (d-2) .
\end{align}
Because
\begin{align}
& d \sqrt{x} \geq 2\sqrt{(1 - x)(d-1)} \sqrt{x} + x (d-2) \nonumber \\
\Longleftrightarrow{} & d \geq 2\sqrt{(1 - x)(d-1)} + \sqrt{x}(d-2) \nonumber \\
\Longleftrightarrow{} & d \geq 2\sqrt{d-1} \sin\phi  + (d-2) \cos\phi \nonumber \\
\Longleftarrow{} & d \geq \sqrt{(2\sqrt{d-1})^2 + (d-2)^2} = d,
\end{align}
\equationref{eq:temp_eq} is proved, and so is \autoref{thm:thm2}.
\end{proof}

The above result is of great importance because the case that the rank-one projectors of two measurement operators form a pair of MUBs is believed to maximize the sum of error and disturbance.
\autoref{thm:thm2} confirms this belief for projective measurements in any finite-dimensional Hilbert space.
Moreover, \autoref{thm:thm2} tells us again when $\mathcal{A}'=\mathcal{A}$, the sum of error and disturbance is minimized.

It is natural to ask if \autoref{thm:thm1} is true for general measurement operations when the Hilbert space dimension $d>2$.
Our numerical simulations find that while \autoref{thm:thm1} does not hold when $d>2$, the theorem seems to be correct by adopting the following slightly relaxed definition of $\varepsilon({\mathcal A},{\mathcal B})$,
\begin{equation}
\varepsilon(\mathcal{A},\mathcal{B}) = \min_{\tilde{\mathcal B}}
\varepsilon(\mathcal{A},\tilde{\mathcal{B}}),
\end{equation}
where the minimum is over all possible projective measurements with rank one projectors $\tilde{\mathcal B}$ corresponding to the measurement operation $\mathcal{B}$.
This relaxed definition of $\varepsilon$ means that we only distinguish between the classical probability distributions resulting from measurements $\mathcal{A}$ and $\mathcal{B}$ with no regard to the labeling of different measurement outcomes.
With this relaxed definition for $\varepsilon$, we make the following conjecture.

\begin{conjecture}[Trade-off in $d$-dimensional Hilbert space]
\label{thm:conj}
For all $d \geq 2$, the following trade-off between error and disturbance holds:
\begin{equation}
\varepsilon(\mathcal{A},\mathcal{A}') + \eta(\mathcal{A}',\mathcal{B}) \geq f(\mathcal{A},\mathcal{B})
\end{equation}
and
\begin{equation}
\Delta(\mathcal{A},\mathcal{A}',\mathcal{B}) \geq f(\mathcal{A},\mathcal{B}),
\end{equation}
where
\begin{equation}
f(\mathcal{A},\mathcal{B})
= \min \{ \varepsilon(\mathcal{A},\mathcal{B}),\eta(\mathcal{A},\mathcal{B})  \}.
\end{equation}
Moreover, the equalities are both attained either when for all $i$, $\ket{a'_i} = \ket{a_i}$, or when for all $i$, $\ket{a'_i} = \ket{b_i}$.
\end{conjecture}

If our conjecture can be proved, the validity of our trade-off relation will be extended to any operator in Hilbert space.
Our sharp lower bound $f(\mathcal{A},\mathcal{B})$ will then serve as a importance physical quantity characterizing the trade-off between error and disturbance in consecutive measurement process.

\section{Conclusion}
\label{sec:conclusion}

To summarize, we have introduced new definitions of error and disturbance in quantum measurement that are both state- and eigenvalue-independent.
An error-disturbance trade-off for a single qubit as well as for a special case of finite-dimensional Hilbert space have been proven.
In particular, it implies that the lower bound of the sum of error and disturbance can be obtained when $\mathcal{A}'=\mathcal{A}$.
Physically, it means it is not worth sacrificing error to lower the disturbance.
Finally, we introduce \conjref{thm:conj} which is a non-trivial extension of \autoref{thm:thm1} to $d$-dimensional Hilbert space with $d>2$.
Our numerical simulations supports this conjecture; and
if it can be proved, an error-disturbance trade-off just like Heisenberg's uncertainty principle can be derived.

\begin{acknowledgments}
S. S. Zhou and S. Wu are supported by the National Natural Science Foundation of China under Grants No.~11275181 and No.~11475084,
and the Fundamental Research Funds for the Central Universities under Grant No.~20620140531.
H.\ F. Chau is supported in part by the RGC Grant HKU~700712P of the Hong Kong SAR Government.
\end{acknowledgments}

%

\end{document}